\documentclass{article}
\usepackage[utf8]{inputenc}
\usepackage{biblatex}
\usepackage{graphics,graphicx}
\bibliography{Bibliography}
\usepackage[T1]{fontenc}
\usepackage[english]{babel}
\usepackage{microtype} 
\usepackage{tabularx} 
\usepackage{booktabs} 

\usepackage{amssymb, amsmath, latexsym, amsfonts, amsthm, mathrsfs} 

\usepackage{fdsymbol} 

\newtheorem{theorem}{Theorem}[section]

\newtheorem{lemma}[theorem]{Lemma}

\newtheorem*{thm*}{Theorem}

\newtheorem*{prop*}{Proposition}

\newtheorem*{conj*}{Conjecture}

\newtheorem*{quest*}{Question}

\providecommand{\keywords}[1]
{
  \small	
  \textbf{\textit{Keywords---}} #1
}

\title{Physical Zero-knowledge Proofs for Flow Free, Hamiltonian Cycles, and Many-to-many k-disjoint Covering Paths}
\author{Eammon Hart and Joshua A. McGinnis}
\date{\today}
\begin{document}
\maketitle

\begin{abstract}
In this paper we describe protocols which use a standard deck of cards to provide a perfectly sound zero-knowledge proof for Hamiltonian cycles and Flow Free puzzles. The latter can easily be extended to provide a protocol for a zero-knowledge proof of many-to-many k-disjoint path coverings. 
\end{abstract}

\keywords{Flow Free, Hamiltonian cycles, disjoint covering paths, ZKP, card-based cryptography, Numberlink, graph, puzzle}

\section{Introduction}
Hamiltonian cycles have 
been a go to example for introductions to zero-knowledge proofs for decades 
\cite{blum1986prove}, 
\cite{Feige1999NonInteractive}, but as far as we can tell all existing proofs are probabilistic. Unlike those methods in this paper we use a deck of cards to provide a perfectly sound zero-knowledge proof for the possession of a Hamiltonian cycle, i.e. the probability of deception is 0.

Flow Free is a popular logic puzzle app game which has over 100,000,000 downloads in the Google play story \cite{google} and whose spin off variants of Flow Free Hexes \cite{google2} and Flow Free Bridges \cite{google3} have millions of downloads as well.  They are variants of the Numberlink puzzle which dates back to at least 1897 when it was published by Sam Loyd.  A similar game was published by Nikoli, the Japanese publisher famous for Sudoku. Flow Free involves an $n \times m$ grid with labeled cells, such that one must connect given pairs of cells with paths that collectively cover all of the cells. The Nikoli variant places the additional requirement of minimizing the turns in the paths connecting the given pairs of cells, but drops the requirement that every cell must be filled. This means that solution paths between connected pair of cells may need to be non-simple, which adds some additional complexity to the problem of finding a zero-knowledge proof. 
\begin{figure}
    \centering
    \includegraphics[scale=.5]{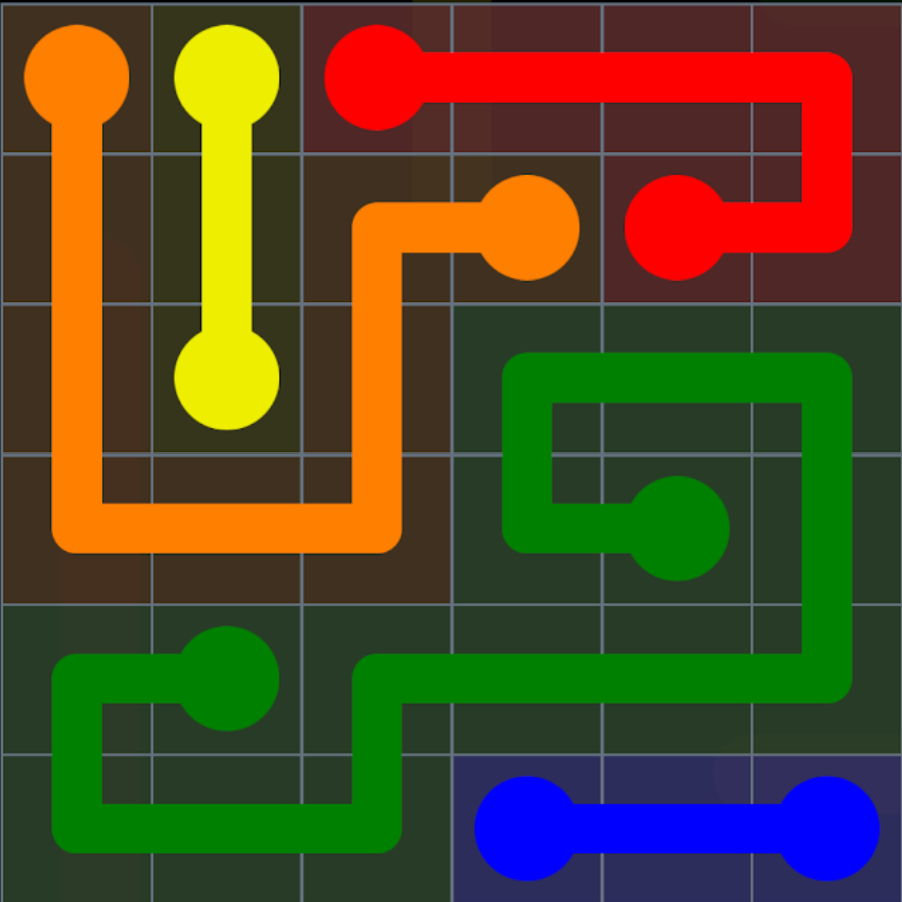}
    \caption{A solved game from the Flow Free app, where the solution contains paths which are non-simple. Both the green path and red path are non-simple.}
    \label{fig:FlowFree}
\end{figure}
Flow Free also relates deeply to research on a generalization of the problem of finding a Hamiltonian path with given starting points called the disjoint covering paths problem (studied by \cite{Park2021ToruslikeGA}, \cite{KRONENTHAL201714}, \cite{ZHANG2013103}, \cite{PARK2015168}, and \cite{MORAN1991179} among many others), because the solution to each game can easily be mapped to a disjoint covering path for the underlying adjacency graph of the game. Numberlink corresponds to the question of vertex-disjoint paths that do not necessarily cover every vertex.  Past study of the game has primarily been from the computer science direction, with solving Flow Free and Zig Zag Numberlink being shown to be NP Complete \cite{AdcockAaron2015ZNiN}, algorithmic explorations enumerating all solutions to a board, identifying games with a unique solution in \cite{a5020176}, and efficiently solving puzzles using Monte Carlo methods \cite{8848043}. \cite{nef2020development} has also recently used it in assessing neurological degeneration. More directly related to our paper \cite{RuangwisesSuthee2020PZPf} provides a physical zero-knowledge proof of solutions to games of Numberlink and more broadly, of knowledge of k vertex-disjoint paths, but they do not provide a zero-knowledge proof for solutions to games Flow Free or the setting of knowledge of  k vertex-disjoint paths which cover the vertices (they explicitly identify this as a potential avenue for future research). In this paper we provided several extensions to the methodology in \cite{RuangwisesSuthee2020PZPf} that allow us to construct perfectly sound zero-knowledge proofs for Hamiltonian cycles, games of Flow Free, and general many-to-many $k$-disjoint path covers. The main obstacle that we overcome is in dealing with non-simple paths, i.e. a vertex in a path is adjacent to another vertex in the same path but the vertices are not path adjacent; see Figure \ref{fig:FlowFree}. It should also be noted that this paper sits in a young, but growing literature of perfectly sound physical zero-knowledge proofs that also includes Zig-zag Numberlink \cite{RuangwisesSuthee2020PZPf}, Ripple Effect\cite{RUANGWISES2021115}, Sudoku \cite{SASAKI2020135} (this differs from previous non-certain zero-knowledge proofs for Sudoku provided by \cite{10.1007/978-3-540-72914-3_16}), Kakuro \cite{Daiki}, Takuzo and Jousan \cite{miyahara_et_al:LIPIcs:2020:12781}, 
 Nurikabe and Hitori \cite{10.1007/978-3-030-80049-9_37}, among others.

\section{Hamiltonian Cycles}

Methods for providing zero-knowledge proofs for Hamiltonian cycles dates back at least to \cite{blum1986prove} and have been elaborated in other settings such as in the non-interactive case in \cite{Feige1999NonInteractive}, but as far we can tell all zero-knowledge proofs for Hamiltonicity that currently exist in the literature are not perfectly sound. For example Blum's protocol, which seems to be the canonical one, gives probability $2^{-n}$ that the prover does not actually know a Hamiltonian cycle, where $n$ is the number of checks performed. Our protocol can therefore be seen as an improvement on the previous ones by providing a zero-knowledge proof with perfect soundness. 

\subsection{The Protocol}

In what follows, cards are used for indexing matrices and encoding information for edges in the graph. The encoding is done via two suits $\varheartsuit$ (hearts) and $\spadesuit$ (spades). Indexing can be done by using the numerical value of the cards (for our purposes, since the number of indexing cards needs to be large, we assume our deck of cards has relatively large values such as the rare $30$ of clubs.) 

We use an encoding system similar to that found in  \cite{RuangwisesSuthee2020PZPf} to encode numbers. Define $E_k(i)$ to be a sequence of $k$ encoding (marking) cards in a row, with the $i^{\text{th}}$ card being a heart and all other cards being spades. If $i = 0$, then let the entire sequence of cards be spades. For examples $E_3(2) = \spadesuit \varheartsuit \spadesuit$, while $E_4(3) = \spadesuit \spadesuit \varheartsuit \spadesuit$. For the purposes of checking Hamiltonicity we only need to use $E_1(1) = \varheartsuit$, $E_1(0) = \spadesuit$, and $E_n(k)$ for all $0 \leq k \leq n$ where $n$ is the number of vertices in the graph. Cards are often laid out in matrices, and a $0^{\text{th}}$ column or row is always used for indexing cards, when they are needed.

One last important remark is that we often employ scramble shuffles which are essentially methods of uniformly randomly permuting certain rows or columns of a matrix of cards. Our variants are directly inspired by the \textit{double scramble shuffle} in \cite{RuangwisesSuthee2020PZPf} which in turn was inspired by the \textit{pile scramble shuffle} in \cite{Ishikawa2015Scramble}, but similar methods have been used in other physical zero-knowledge proofs in the introduction. Both papers suggest that the employment of envelopes should make such a process physically possible.

\subsubsection{Deployment of Cards}
The first couple of steps explain how a prover encodes a solution. 

\textbf{Step 1}: The prover should privately mark every edge in the Hamiltonian cycle with $E_1(1)$ and every other edge with $E_1(0)$, i.e. they place a single heart face down on every edge that is part of the cycle and a single spade face down on every edge that is not part of the cycle.

\textbf{Step 2}: The prover privately marks every edge with a $E_n(k)$ where $k \in {0,1,\dots,n}$. Physically this means the prover places a face down stack of $n$ cards encoding some number $k$ onto each edge of the graph. This is done in the following way: the prover arbitrarily picks one edge in the cycle and marks it with $E_n(1)$ and marks one of the two edges adjacent to it and also in the cycle with $E_n(2)$. The prover marks the one edge in the cycle which is unmarked and adjacent to $E_n(2)$ with $E_n(3).$ They continue this way, marking adjacent edges of the cycle with consecutive encoded numbers until the last marker $E_n(n)$ is placed. For an actual Hamiltonian cycle, this marker will necessarily be placed in the last unmarked edge of the cycle, next to $E_n(1)$. The edges not in the cycle should be marked with $E_n(0).$  
\subsubsection{Checking the Solution}
The prover and verifier go through a process of checking every vertex publicly.

\textbf{Step 3}: The cards on every edge connected to the vertex currently being checked are used to create a matrix. In column 0 a series of indexing cards $\clubsuit 1, \clubsuit 2, \ldots \clubsuit d$ are placed where $d$ is the degree of the vertex being checked. Now for each edge connected to the vertex that is being checked, in column 1 place face down the single cards, $E_1(1)$ or $E_1(0),$ encoding whether that edge is a part of the Hamiltonian cycle or not. In column 2, place face down the $E_n(k)$ markers, i.e. the stack of cards encoding each edge's number (See Figure \ref{fig:Matrix1}.) Flip the indexing cards face down so that all cards in the matrix are now face down.

\textbf{Step 4}: Perform a row scramble shuffle, which is defined by randomly permuting the $d$ rows of the matrix. Then flip over all cards in column 1. If there are not 2 $ \varheartsuit$'s in the column, the verifier rejects the solution outright and the protocol terminates. If there are 2 $ \varheartsuit$'s in the column, proceed.

\textbf{Step 5}: Call the two rows with hearts row $i$ and row $j$. Take the stacks in column 2 of rows $i$ and $j$ and use them to create a second matrix of cards as follows. Column $0$ should be a column of face down indexing cards $1,2 , \ldots,n$. Column 1 is the stack of cards from the second column of the $i^{\text{th}}$ row of the first matrix. These cards should be laid out face down, one by one, preserving the order meaning that a card in the $k^{\text{th}}$ spot would be placed in the same row as the $k^{\text{th}}$ indexing card. Column 2 is filled the same way but with the stack from the second column and $j^{\text{th}}$ row of the first matrix.  

\textbf{Step 6}: Shift every card of column 2 of the second matrix down one row with the $n^{\text{th}}$ card cycling back to the first row. (See Figure \ref{Matrix2}.)

\textbf{Step 7} : Perform a row scramble shuffle. Flip over all cards in column 1 and 2. If there are 2 hearts in the same row then the verifier should accept this vertex. Either way, turn all cards face down, perform another row scramble shuffle, then turn over the indexing cards in column $0$. Put the rows in correct order using the indexing cards. Then shift all the cards in column 2 up one row, with the card in the first row going to the $n^{\text{th}}$ row.

 \textbf{Step 8}: If the verifier has accepted the vertex, go to step 9. If the verifier has not accepted the vertex, shift every card in column 2 up by 1 row with the card in the first row going to the last row. Turn the indexing cards face down, and perform a row scramble shuffle on the matrix. Flip over the cards in column 1 and 2. If there are not 2 $\varheartsuit$'s in the same row, the verifier should reject the solution entirely and the protocol terminates. If there are 2 $\varheartsuit$'s, the verifier should accept the vertex. Turn all cards face down, perform the row scramble shuffle, flip over column 0 and use the indexing cards to place the rows in their original places, shift all cards in column 2 down a row, with the card in the last row, going to the first row.
 
 \textbf{Step 9}: Reform the stack of cards from which column 1 of the second matrix was made and return this stack to the second column and $i^{\text{th}}$ row of the first matrix. Likewise, reform the stack of cards from which second column of the second matrix was made and return it to second column and  $j^{\text{th}}$ row the first matrix. Now there is only the first matrix. Turn any cards face down and perform a row scramble shuffle.  Finally, flip over the indexing cards in column 0 and use them to re-deploy the marking cards onto the edges from which they came.

\textbf{Step 10}:
Repeat steps 3 through 9 for all vertices on the graph. 

\textbf{Step 11}:
Once the verifier has accepted every vertex, place each of the stacks marking the $E_n(k)$ from every edge of the graph into a single column. The column should have the same length as the number of edges in the graph. Perform a row scramble shuffle. Now reveal all stacks of cards and check that there exists exactly 1 $E_n(i)$ for all $1 \leq i \leq n$. If this is true then the verifier should accept the solution. Otherwise the verifier should reject it.

\begin{figure}[ht]
    \centering
    \includegraphics[width=.5\textwidth]{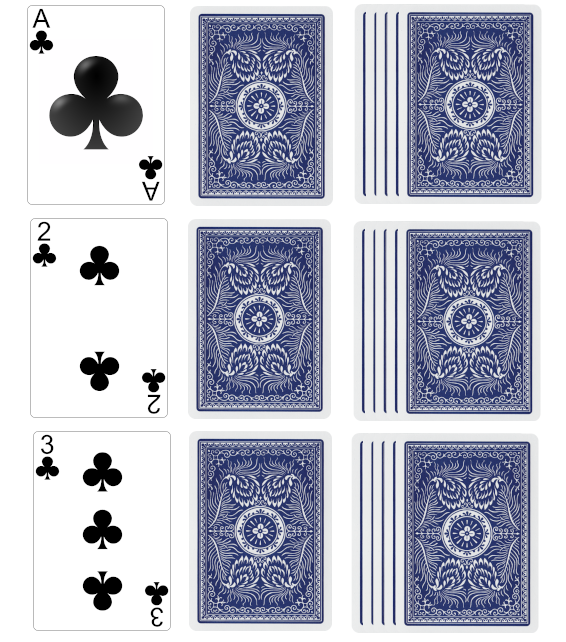}
    \caption{This is an example of a matrix from step 3 of the Hamiltonian Cycle protocol. The $0^{\text{th}}$ column are the indexing cards. The $1^{\text{st}}$ column is a column of single cards pulled from all adjacent edges of the vertex that is currently being checked, which had degree 3. The $2^{\text{nd}}$ column contains the stack of cards from each of the adjacent edges. In this case the graph has $5$ vertices, so each stack of cards has 5 cards.}
    \label{fig:Matrix1}
\end{figure}

\begin{figure}[ht]
    \centering
    \includegraphics{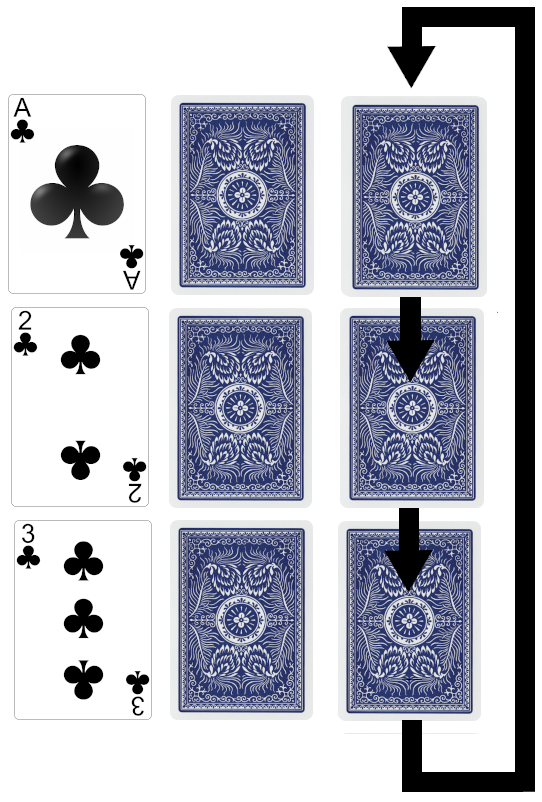}
    \caption{This is an example of a second matrix from step 6 in which it is checked that adjacent edges in the Hamiltonian cycle are encoded with adjacent numbers modulo the number of vertices, which in this case is $3$. (Note this example does not coincide with the example in Figure \ref{fig:Matrix1}, since in that case, the second matrix would have 5 rows.) }
    \label{Matrix2}
\end{figure}

\subsection{Explaining the Steps}

We have split the protocol into steps so as to explain intuitively the reasons behind each step.

\textbf{Step 1 Reason}: The prover needs some way to encode their solutions into the graph, if all the cards were turned face up after this step, one would clearly see the Hamiltonian cycle.

\textbf{Step 2 Reason}: The prover now introduces some seemingly extraneous information to their solution, but this extra information allows the verifier to confirm that the prover has one cycle and not multiple, as we shall see below.

After the first two steps, the prover has privately marked the graph with the solution and the rest of the protocol is done publicly, so this is a good time to introduce the intuition of  the protocol. 

\emph{1. Every vertex must be connected to exactly two edges which are in a cycle.}

\emph{2. A Hamiltonian cycle must contain exactly $n$ edges which occur sequentially one after the other.}

The first condition guarantees that any vertex is apart of a cycle. Two disjoint cycles which cover the graph still satisfy the first condition, but not the second. The protocol is then designed to convey that these first two conditions are met, without revealing the Hamiltonian cycle itself. 

\textbf{Step 3 Reason}: The remainder of the steps are done publicly. Cards are mostly placed face down so as not to reveal too much information to the verifier, but the creation of the matrices should be done publicly so that neither the verifier nor prover can tamper with the cards that prover has submitted as a solution.

\textbf{Step 4 Reason}: This step checks the first condition of the intuition, that every vertex is connected to exactly two edges in the cycle. The scramble shuffle is employed so that the verifier cannot see which edge's encoding cards belong to which columns. 

\textbf{Step 5 Reason}:

The second condition still needs to be checked so the other set of marking cards, the card encoding the edges' numbers, must be checked. This check is prepared by laying out the two stacks into two different columns of a new matrix. Thus the two columns now encode the two numbers associated with the two edges in the Hamiltonian cycle connected to the checked vertex. 

\textbf{Step 6 Reason}:
The idea behind this step is that the numbers marking adjacent edges should have a difference of $1$ modulo $n$.  Therefore there are two possibilities. Either the first column of the second matrix encodes the larger number or the second column does. We start by checking the former, but in order to shuffle the columns and preserve the adjacency of the two numbers, we first attempt to make the numbers equal. If the second column is the smaller number, step 6 makes both columns encode the same number.

\textbf{Step 7 Reason}: If, in step 6, it was indeed true that the second column originally encoded the smaller number, the verifier should see that the two columns encode the same number after the second column was shifted down.  A row scrammble shuffle allows the verifier to check this without the prover revealing the numbers themselves. After this check the cards should be put back in the positions in which they started, in the second matrix.

\textbf{Step 8 Reason}:
In Step 7 Reason, we assumed the second column encoded the smaller number, but if this turns out to be false, it could be the first column encodes the smaller number, so the second column is now shifted up, and a similar check to the one in the previous step is preformed. There is a chance here that the verifier rejects the solution entirely if they do not see adjacent numbers. This could happen if the prover actually is trying to cheat and uses two cycles, thus perhaps must skip numbers along the way, because they do not want to be caught reusing numbers; see the reason for step 10. 

\textbf{Step 9 Reason}:
Step 9 is simply a step for taking all the cards that have been placed in matrices and putting them back onto the graph the way the prover had first set them on the graph in steps 1 and 2.

\textbf{Step 10 Reason}: Every vertex needs to under go the same check.

\textbf{Step 11 Reason}: The final step is a global check. The verifier just needs to check that the prover has not encoded multiple edges with the same number; see Figure \ref{disjoint}. We should note here that the vast majority of the cards used are actually "decoys" encoding $E_n(0)$ and placed on edges which are not apart of the Hamiltonian cycle.

\begin{figure}
    \centering
    \includegraphics[width=\textwidth]{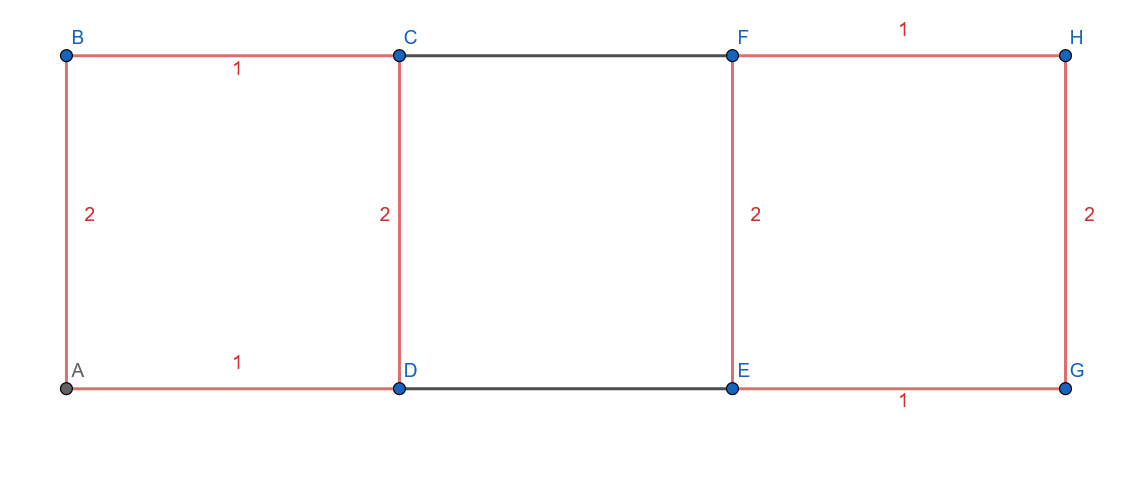}
    \caption{The graph above is showing the encoding of a non-Hamiltonian cycle that would escape local vertex checks. In the example, the prover has encoded two cycles with alternating $E_{8}(1)$'s and $E_{8}(2)$'s. The global check however would reveal that $E_{8}(1)$ and $E_{8}(2)$ occur multiple times and no other $E_{8}(k)$ occurs for $k\neq 0,1,2$. Thus the verifier would reject.}
    \label{disjoint}
\end{figure}

\subsection{Proofs of Security}

\begin{lemma}{Perfect Completeness:}
If the prover has a Hamiltonian cycle, then the verifier always accepts the solution.
\end{lemma}

\begin{proof}
If the prover provides an actual Hamiltonian Cycle and deploys cards in the manner described in the first two steps, then every vertex has exactly two edges that are marked with a heart, and therefore the verifier approves step 4. If the prover provides an actual Hamiltonian cycle, then the two edges adjacent on any given vertex are encoded by consecutive numbers mod $n$, i.e. the stacks in the second column of rows $i$ and $j$ are $E_n(k)$ and $E_n(k-1)$ or $E_n(k-1)$ and $E_n(k)$. If the first is the case, then the shift down in step 6 puts the hearts for the encoding cards in the same row and the verifier accepts on step 7. If the latter is the case, then once everything has been put back in its original place, the shift up in step 8 puts the two hearts in the same column and the verifier accepts in step 8.

If the prover has provides an actual Hamiltonian cycle the verifier accepts for every vertex, and the protocol moves on to step 11. If the prover provides a valid Hamiltonian cycle, then there is exactly 1 $E_n(i)$ for every $i \in \{1, \dots , n\}$ and the verifier accepts on step 11 and therefore accepts the whole solution.
\end{proof}

\begin{lemma}{Perfect Soundness:}
If the prover does not know a Hamiltonian Cycle, then the verifier always rejects.
\end{lemma}

\begin{proof}
We prove the contrapositive. Assume that the verifier accepts. Pick an arbitrary vertex $v_1$. We say that two vertices are path adjacent, if the prover has marked the edge between them with a heart. The verification process confirms that each vertex, is path adjacent to two other vertices, so let $v_{2}$ be path adjacent to $v_1$ and in turn $v_{3}$ to  $v_{2}$ and so on, so that we may list path adjacencies like 
\[v_{1},v_{2},v_{3}, \dots\].

Since the graph has a finite number of vertices, this either terminates, meaning we get to a vertex with only one path adjacent vertex, or it goes in a cycle. The first case cannot happen if the verifier has accepted, because a vertex without exactly 2 path adjacencies would not be accepted. Thus this is a cycle provided by the prover. However, we still need to show that every vertex appears in the cycle.

 Note that edges must be encoded with numbers in $\{0,1,2,\dots,n\}$ with numbers in $\{1, \dots ,n\}$ being used exactly once. The global check, at the end of the verification process, ensures that this is the case. Let us consider the edges from the cycle above as
\[e_{1,2}, e_{2,3}, \dots.\] The verification process also checks that these edges are encoded with consecutive numbers modulo $n$. Thus this cycle cannot have less than $n$ edges because that would require the prover to use a number in $\{1,\dots n\}$ multiple times or for the some edges to not be encoded with consecutive numbers. Therefore the cycle indeed contains all the vertices.

\end{proof}

\begin{lemma}{Zero-knowledge:}
During the verification process, the verifier learns nothing about the prover's solution.
\end{lemma}

\begin{proof}
The backs of the cards are identical, so the only time that the verifier can learn any information about the solution is when the cards are face side up. Therefore the only steps in which the verifier could potentially get information about the solution are 4, 7, 8 and 11. 

The verifier cannot get information at step 4 because of the shuffle in step 3. The probability of the two hearts landing in a pair of rows is uniformly distributed, and thus, from the verifier's perspective, all pairs of edges are equally likely to be the edges of the Hamiltonian cycle. 

Similarly, the shuffles prevent the verifier from getting information at steps 7, 8, or 11. The shuffles hide the number with which each edge is encoded by making all numbers uniformly probable. 

Finally note that index cards are only turned face up after all cards have been turned face down and shuffled, thus preventing the verifier gleaning any information from them.
\end{proof}

(If you are trying this with an actual deck of cards, the added information of the value on the marking cards as opposed to suit does leak information to the verifier, but if you are dealing with a stripped down system where the suit is only information that is conveyed, say by making every marking card an ace, then no information is leaked.)

\section{Flow Free}
The rules of Flow Free on are simple. To explain, we call each square in a Flow Free grid a cell. We call the edges between the cells walls. In a rectangular grid, all cells besides the ones at the corner have four walls, and each wall is the wall of two different cells. Two cells are adjacent to one another if they share a wall. Two walls are adjacent to one another if they are walls of the same cell. There is a special type of cell, called a terminal cell. Terminal cells are given at the outset of the game and come in pairs. The pairs are labeled with colors or numbers and the player's goal is to connect the cells within all pairs via a path. A path contains walls and cells through which it passes. In a solution, walls and cells are only every contained in at most one path. Two walls are path adjacent if they are adjacent and contained in the same path. A solution path for a pair of terminal cells must thus pass through adjacent walls starting at the wall of one of the terminal cells and ending at the wall of the other. Finally every cell must be contained in exactly one solution path; see Figure \ref{fig:FlowFree}.

\subsection{The Protocol}

\subsubsection{The Prover's Card Deployment}
 Let $k$ be the number of pairs of given terminal cells in an $m \times n$ grid. First the prover should privately place stacks of cards on each wall in the manner described below. Once the solution is encoded, there will be two stacks of cards on each wall, so call this first stack, type I. Type I stacks encode the $E_{k}(j)$. $E_k(j)$ should be placed on every wall through which the path between the $j^{\text{th}}$ pair of terminal cells passes, and if no path passes through a given wall then the prover should place $E_k(0)$. The prover then needs to privately place a second stack of cards on each wall. These are type II stacks, which encode the $E_{(n \times m) -k}(i)$. For walls that are not in a solution path, the prover places $E_{(n \times m)-k}(0)$. The prover should also pick a wall of a terminal cell which is in a path, onto which, they place a second stack $E_{(n \times m)-k}(1)$. On the wall that is path adjacent to this one, the prover should place $E_{(n \times m)-k}(2)$, and so on until the path path terminates. On the final wall in this path, which in a correct solution is the wall of the another terminal cell, the prover places $E_{(n \times m)-k}(i)$, where $i$ would be the number of walls contained in the path. For the next terminal cell wall of a path not yet encoded by a type II stack of cards, the prover performs the same steps as were done on the previous path but starts by encoding the first wall with $E_{(n \times m)-k}(i+1).$ The second wall should be encoded with $E_{(n \times m)-k}(i+2)$ and so on. This process is repeated until all walls are encoded by two stacks of cards. We note that the prover can do this by putting down the stacks of cards in any order so as not to accidentally give away anything by doing the steps above in the order they are written i.e. this is done privately. See Figure \ref{fig:deployed_cards}.
 
 \begin{figure}
     \centering
     \includegraphics[width=\textwidth]{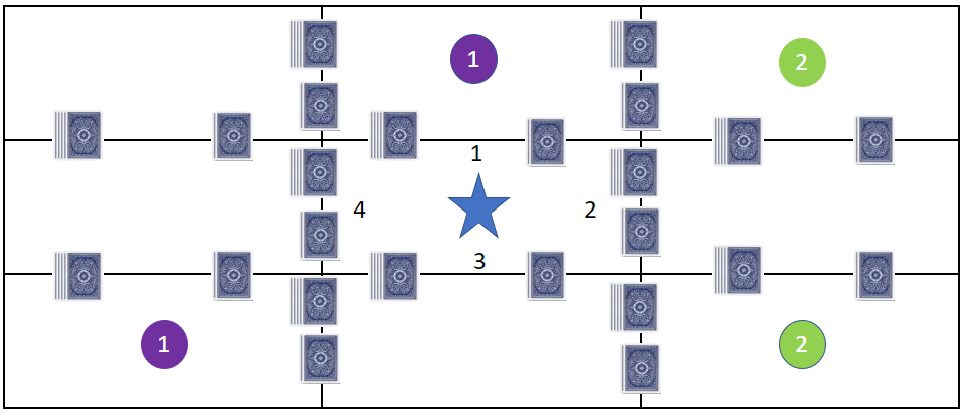}
     \caption{This is an example of a board with cards deployed. On each wall, two stacks of cards have been placed. In this cases, the smaller stack has only two cards, since $k=2$. The larger stack has seven cards which comes from $m \times n -k=3 \times 3 -2$. The star in the center indicates the cell currently being checked.}
     \label{fig:deployed_cards}
 \end{figure}

\subsubsection{Checking the Terminal Cells}

First we must check the terminal cells publicly. It is already public knowledge that this cell corresponds to the $i^{\text{th}}$ path. 

\textbf{Step 1}: Take the type I stacks from all walls that are adjacent to a given terminal cell and use them to construct a matrix. Suppose it is a terminal cell from the $i^{\text{th}}$ pair. Each of the rows is constructed by placing the cards from each of the stacks face down and maintaining the order of the stack meaning a card in the $i^{\text{th}}$ position of the stack ends up in the $i^{\text{th}}$ column of the matrix. Each stack should have its own row and indexing cards should be placed in the $0^{\text{th}}$ column  to keep track of which row corresponds to which wall of the terminal cell. 

\textbf{Step 2}: Turn the indexing cards face down and perform a row scramble shuffle. 

\textbf{Step 3}: Turn over the the cards in columns 1 through $k$. If the verifier sees exactly 1 $\varheartsuit$ in the $i^{\text{th}}$ column and no $\varheartsuit$'s anywhere else, the verifier accepts this cell. Otherwise the verifier rejects the whole solution. 

\textbf{Step 4}: Turn all cards face down perform a row scramble shuffle, flip over the indexing cards and use them to redeploy the stacks of cards on the board. 

\textbf{Step 5}: The verifier repeats the process for all terminal cells. If the verifier accepts for all terminal cells, then move on to checking the non-terminal cells.

\subsubsection{Checking the Non-terminal Cells}
Next the non-terminal cells are checked.

\textbf{Step 1}: Pick a non-terminal cell. Construct a matrix such that row 0 and column 0 are indexing cards. Now take the type I stack from one of the adjacent walls, placing the cards from this stack face down to form row 1. As usual ordering should be maintained so that a card in the $i^{\text{th}}$ position in the stack ends up in the $i^{\text{th}}$ column. Then place the entire type II stack in its own entry in the first row at the end, in column $k+1$. Repeat this process for each adjacent cell wall, adding a new row to the matrix each time. Let $r$ be the number of walls associated with this cell. The largest row index is thus $r$. See \ref{fig:Matrix3}.

\textbf{Step 2}: Turn all indexing cards face down and randomly permute rows $1$ though $r$. Then randomly permute the columns $1$ through $k$. 

\textbf{Step 3}: Flip over all the cards that are in both columns $1$ through $k$ and rows $1$ through $r$. If there are exactly 2 $\varheartsuit$'s in the same column and no other $\varheartsuit$'s, the verifier should accept and move to the next step. Otherwise the verifier should reject the solution.

\textbf{Step 4}: If the verifier accepts the previous step, then take  the type II stacks in the $k + 1$ column from the two rows containing the hearts and create a second matrix by deploying them in two columns in the same way it was done in step 5 of the Hamiltonian cycle protocol, but note, the number of rows in the new matrix is $n \times m -k$ instead of $n$.

\textbf{Step 5}: Perform steps 6-8 of the cycle protocol including the relevant acceptances or rejections.

\textbf{Step 6}: If the solution is not rejected in step 5, turn the columns of the second matrix back into the original stacks and put the stacks back into the first matrix, just as was done in the cycle checking protocol

\textbf{Step 7}: Turn all cards face down and randomly permute rows $1, 2, \ldots, r$ and and then the columns $1, 2, \ldots, k$. 

\textbf{Step 8}: Flip over the indexing cards in row 0 and use this to put the columns back in the right order.

\textbf{Step 9}: Flip over the indexing cards in column 0 and use this to redeploy the stacks of cards onto the correct walls as they originally were.

\textbf{Step 10}: Repeat the process for every non-terminal cell. If it is accepted for every cell, move to checking the global properties.

\begin{figure}
    \centering
    \includegraphics[width=.5\textwidth]{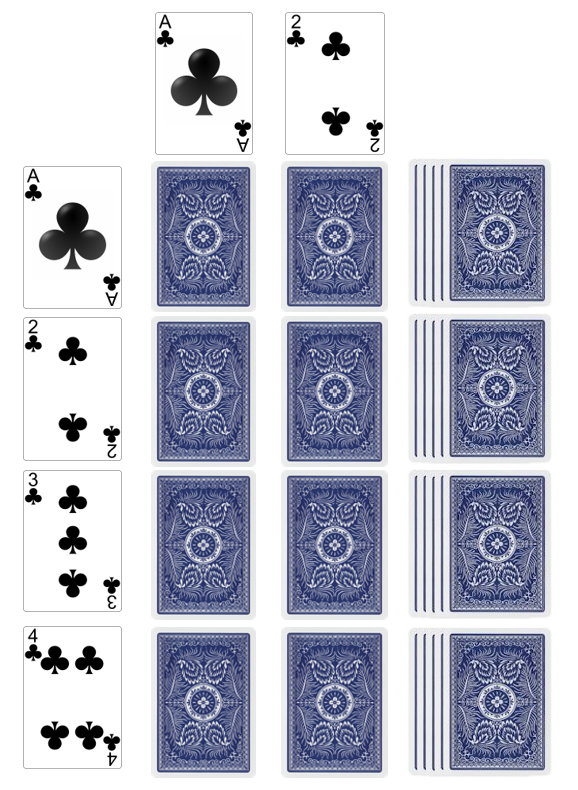}
    \caption{This is the first matrix of cards that should be formed in checking the non-terminal cell marked with a star in Figure \ref{fig:deployed_cards}. The type I stacks from the four adjacent walls have been laid out so that the $i^{th}$ card in the stack from $j^{th}$ wall is in the $(i,j)$ position. The last column contains the type II stacks. (There appear to be only five cards in each of the type II stacks, but in the example there would be $3 \times 3 -2=7$.)  }
    \label{fig:Matrix3}
\end{figure}

\subsubsection{Checking global properties}
A final global check is needed.

\textbf{Step 1}: Take the type II stacks from every wall in the puzzle. Line them up face down so that each stack is laid out in a column and perform a column scramble shuffle. 

\textbf{Step 2}: Turn all cards face up and check that there is exactly 1 heart in each row. If there is, then the verifier should accept the solution to the whole puzzle. If not, the verifier should reject the solution. 

\subsection{Proof of Security}

\begin{lemma}{Perfect Completeness:}
If the prover has a solution to a Flow Free Puzzle, the verifier always accepts the solution.
\end{lemma}

\begin{proof}
If the prover provides a card deployment that corresponds to a solution, then there is exactly 1 $\varheartsuit$ in the right column in step 3 of each terminal cell check. Therefore the verifier accepts the terminal cells and moves on to the non-terminal cells. 

Each non-terminal cell must have 2 adjacent walls with stacks of type I that have $\varheartsuit$'s in the same column at deployment, while the other walls should have stacks of type I with no heart. These facts are not changed by either the random permutation of the rows or columns, therefore the verifier accepts step 3 of the non-terminal cell check. The two stacks of type II, corresponding to the walls with the stacks of type I with hearts, should encode adjacent numbers i.e. $E_{m \times n -k}(j)$ and $E_{m \times n -k }(j+1)$ therefore the verifier accepts during step 5 for all non-terminal cells.

Checking the global properties, the verifier accepts in step 2 because the prover has deployed one stack $E_{n \times m -k}(j)$ for each $j \in \{1, \ldots n\times m - k\}$, since this is exactly the number of walls that are occupied by a path, in a valid solution.  Therefore if presented with a valid solution the verifier always accepts. 

\end{proof}

\begin{lemma}{Perfect Soundness:}
If the prover does not know a solution to a Flow Free Puzzle, then the verifier always rejects.
\end{lemma}

\begin{proof}
We prove the contrapositive. Assume that the verifier accepts. A proper solution has two main objectives: connect each pair of terminal points via a path and fill the grid so that each cell is contained in exactly one path which terminates at terminal cells.

We say a terminal cell belongs to a path if it is one of the terminal cells of the $i^{\text{th}}$ pair,  and it has exactly one wall in a path, and that wall is in path $i$. The check on terminal cells confirms that all terminal cells belong to the correct path. We say that a non-terminal cell belongs to a path $i$ if it has exactly two walls in a path $i$, and its other walls are not in a path. The check on non-terminal cells, confirms that each non-terminal cell belongs to a possibly non-terminating path. These facts guarantee that all paths that do terminate, do so at terminal cells. If a path terminates at a non-terminal cell, the non-terminal cell would not get past the verification process. Furthermore, each terminal cell is in a path that terminates at another terminal cell in the same path. 

If a path does not terminate, it must be a cycle and it must not contain terminal cells. Therefore any cycle has strictly less than $n \times m-k$ walls. However the verification process checks that path adjacent walls are encoded with consecutive numbers modulo $n \times m -k$, so such a cycle must repeat a number e.g. $1, 2, 3, 2,...$, but the global check would reject repeated numbers, so there must not be any cycles. Thus the prover does indeed have a valid solution, since this shows that all non-terminal cells belong to paths that terminate at the correct terminal cells.

(Note, the underlying graph structure (walls = edges, cells =vertexes) of any Flow Free game is simple graph so it is impossible to have a 2 cycle.)
\end{proof}

\begin{lemma}{Zero-knowledge:}
During the verification process, the verifier learns nothing about the prover's solution.
\end{lemma}
 
\begin{proof}

The only times that the verifier can learn about prover's solution is when the cards are flipped face up, i.e. at steps 3 and 4 of the terminal cell check, steps 3, 5, 8, and 9 of the non-terminal cell check and step 2 of the global check. The revealing of indexing cards never gives information because a random permutation with all cards face down is always performed before hand.

Step 3 of the terminal cell check gives no retrievable information about which wall is used because of the random permutation during step 2. The information about which path it corresponds to is retrievable, but that is already public information so no information is given. 

Step 3 of the non-terminal cells gives no retrievable information because of the random permutations in step 2. Step 5 gives no retrievable information for the reasons provided in the Hamiltonian graph section. Steps 8 and 9 give no retrievable information because the random permutations in step 7 strips the information provided in step 3.

Step 2 of the global check gives no information because the random permutation in step 1 strips the information of which stack corresponds to which wall.

\end{proof}

\section{Many-to-Many k-Disjoint Path Cover and Conclusion}

The many-to-many k-disjoint path cover problem studied in \cite{KRONENTHAL201714},\cite{ParkJH2006Mdpc},  \cite{PARK2015168}, \cite{Park2021ToruslikeGA}, \cite{ZHANG2013103}, looks at breaking a graph into $k$ disjoint paths, with given starting and ending vertices (sources and sinks) such that every vertex is in exactly one path. We say it is paired if each source is required to be in the same path as a specified sink, otherwise it is unpaired. Finding a Hamiltonian path with given starting points is a special case of the many-to-many k-disjoint path cover problem, where $k = 1$. Flow Free puzzles are also special cases of the problem if we look at the underlying adjacency graph of the board. If we simply replace cell with vertex and wall with edge, our protocol generalizes to providing a zero-knowledge proof for the paired many-to-many k-disjoint covering path problem for simple graphs and we can also handle the unpaired many-to-many k-disjoint covering path problem by including a column scramble when checking all of the sink vertices.

The slight variation for Flow Free rules meant that the protocol described in \cite{RuangwisesSuthee2020PZPf} could not be neatly extended and so substantial changes were needed as well as an introduction of seemingly extraneous information. This mean that our protocol required orders of magnitude more cards. It uses $\Theta((mn)^2)$  cards on an $m \times n$ grid with $k$ pairs of terminal cells. The $k$ does not appear since although a stack of $k$ cards is needed on each wall, a stack of $m\times n -k$ is needed also needed on each wall. Generalizing this to graphs we would have $\Theta(|E| |V|)$ cards for both the Hamiltonian cycle and many-to-many $k$-disjoint covering paths. A fruitful direction for future research may be seeing if there exists a method that can achieve the same perfect soundness while using fewer cards. Furthermore, our checking algorithm requires many steps, and an algorithm that can check more efficiently would be another possible direction for improvement.  Finally, we encourage generalizations of these to protocols to directed graphs, other games, logic puzzles, and other problems in algorithmic graph theory. 
 
\printbibliography

\end{document}